%% file: main.tex
\documentclass[english]{article}
\usepackage[latin9]{inputenc}
\usepackage{geometry}
\usepackage{amsmath}
\usepackage{amsthm}

\makeatletter
\numberwithin{equation}{section}
\numberwithin{figure}{section}
\theoremstyle{plain}
\newtheorem{thm}{\protect\theoremname}
  \theoremstyle{definition}
  \newtheorem{defn}[thm]{\protect\definitionname}
  \theoremstyle{plain}
  \newtheorem{lem}[thm]{\protect\lemmaname}

\def\final{0}  

\ifnum\final=0  
\newcommand{\note}[1]{[{\tiny \bf #1}]\marginpar{*}}
\newcommand{\sidecomment}[1]{\marginpar{\tiny #1}}
\else 
\newcommand{\note}[1]{}
\newcommand{\sidecomment}[1]{}
\fi  

\usepackage[lined,boxed,ruled,norelsize,algo2e]{algorithm2e}

\usepackage{amsfonts}
\usepackage{tikz}
\usetikzlibrary{shapes,arrows,arrows.meta}

\makeatletter
\newcommand{\xleftrightarrow}[2][]{\ext@arrow 3359\leftrightarrowfill@{#1}{#2}}
\newcommand{\xdashrightarrow}[2][]{\ext@arrow 0359\rightarrowfill@@{#1}{#2}}
\newcommand{\xdashleftarrow}[2][]{\ext@arrow 3095\leftarrowfill@@{#1}{#2}}
\newcommand{\xdashleftrightarrow}[2][]{\ext@arrow 3359\leftrightarrowfill@@{#1}{#2}}
\def\rightarrowfill@@{\arrowfill@@\relax\relbar\rightarrow}
\def\leftarrowfill@@{\arrowfill@@\leftarrow\relbar\relax}
\def\leftrightarrowfill@@{\arrowfill@@\leftarrow\relbar\rightarrow}
\def\arrowfill@@#1#2#3#4{%
  $\m@th\thickmuskip0mu\medmuskip\thickmuskip\thinmuskip\thickmuskip
   \relax#4#1
   \xleaders\hbox{$#4#2$}\hfill
   #3$%
}

\pgfarrowsdeclare{tonew}{tonew}
{
  \ifdim\pgfgetarrowoptions{tonew}=-1pt%
    \pgfutil@tempdima=0.84pt%
    \advance\pgfutil@tempdima by1.3\pgflinewidth%
    \pgfutil@tempdimb=0.21pt%
    \advance\pgfutil@tempdimb by.625\pgflinewidth%
  \else%
    \pgfutil@tempdima=\pgfgetarrowoptions{tonew}%
    \pgfarrowsleftextend{+-0.8\pgfutil@tempdima}%
    \pgfarrowsrightextend{+0.2\pgfutil@tempdima}%
  \fi%
}
{
  \ifdim\pgfgetarrowoptions{tonew}=-1pt%
    \pgfutil@tempdima=0.28pt%
    \advance\pgfutil@tempdima by.3\pgflinewidth%
    \pgfutil@tempdimb=0pt,%
  \else%
    \pgfutil@tempdima=\pgfgetarrowoptions{tonew}%
    \multiply\pgfutil@tempdima by 100%
    \divide\pgfutil@tempdima by 375%
    \pgfutil@tempdimb=0.4\pgflinewidth%
  \fi%
  \pgfsetdash{}{+0pt}
  \pgfsetroundcap
  \pgfsetroundjoin
  \pgfpathmoveto{\pgfpointorigin}
  \pgflineto{\pgfpointadd{\pgfpoint{0.75\pgfutil@tempdima}{0bp}}
                         {\pgfqpoint{-2\pgfutil@tempdimb}{0bp}}}
  \pgfusepathqstroke
  \pgfsetlinewidth{0.8\pgflinewidth}
  \pgfpathmoveto{\pgfpointadd{\pgfqpoint{-3\pgfutil@tempdima}{4\pgfutil@tempdima}}
                             {\pgfqpoint{\pgfutil@tempdimb}{0bp}}}
  \pgfpathcurveto
  {\pgfpointadd{\pgfqpoint{-2.75\pgfutil@tempdima}{2.5\pgfutil@tempdima}}
               {\pgfqpoint{0.5\pgfutil@tempdimb}{0bp}}}
  {\pgfpointadd{\pgfqpoint{0pt}{0.25\pgfutil@tempdima}}
               {\pgfqpoint{-0.5\pgfutil@tempdimb}{0bp}}}
  {\pgfpointadd{\pgfqpoint{0.75\pgfutil@tempdima}{0pt}}
               {\pgfqpoint{-\pgfutil@tempdimb}{0bp}}}
  \pgfpathcurveto
  {\pgfpointadd{\pgfqpoint{0pt}{-0.25\pgfutil@tempdima}}
               {\pgfqpoint{-0.5\pgfutil@tempdimb}{0bp}}}
  {\pgfpointadd{\pgfqpoint{-2.75\pgfutil@tempdima}{-2.5\pgfutil@tempdima}}
               {\pgfqpoint{0.5\pgfutil@tempdimb}{0bp}}}
  {\pgfpointadd{\pgfqpoint{-3\pgfutil@tempdima}{-4\pgfutil@tempdima}}
               {\pgfqpoint{\pgfutil@tempdimb}{0bp}}}
  \pgfusepathqstroke
}

\pgfarrowsdeclarealias{<new}{>new}{tonew}{tonew}
\pgfsetarrowoptions{tonew}{-1pt}
\pgfkeys{/tikz/.cd, arrowhead/.default=-1pt, arrowhead/.code={
  \pgfsetarrowoptions{tonew}{#1},
}}
\makeatother

\makeatother

\usepackage{babel}
  \providecommand{\definitionname}{Definition}
  \providecommand{\lemmaname}{Lemma}
\providecommand{\theoremname}{Theorem}

\begin{document}
\include{general_notation}

\title{Efficient Convex Optimization with Membership Oracles}

\author{Yin Tat Lee\thanks{Microsoft Research, yile@microsoft.com} \and
Aaron Sidford\thanks{Stanford University, sidford@stanford.edu} \and
Santosh S. Vempala \thanks{Georgia Tech, vempala@gatech.edu}}
\maketitle
\begin{abstract}
We consider the problem of minimizing a convex function over a convex
set given access only to an evaluation oracle for the function and
a membership oracle for the set. We give a simple algorithm which
solves this problem with $\otilde(n^{2})$ oracle calls and $\otilde(n^{3})$
additional arithmetic operations. Using this result, we obtain more
efficient reductions among the five basic oracles for convex sets
and functions defined by Grötschel, Lovasz and Schrijver \cite{GLS}.
\end{abstract}
\thispagestyle{empty}\newpage
\clearpage \setcounter{page}{1}

\input{intro.tex}

\input{prelim.tex}

\input{mem_to_sep.tex}

\input{sep_to_opt.tex}

\input{equivalency.tex}

\section*{Acknowledgments}

The authors thank Sébastien Bubeck, Ben Cousins, Sham M. Kakade and
Ravi Kannan for helpful discussions, and Yan Kit Chim for making the
illustrations. 

\bibliographystyle{plain}
\bibliography{acg}

\end{document}

%% file: general_notation.tex
\global\long\def\defeq{\stackrel{\mathrm{{\scriptscriptstyle def}}}{=}}
\global\long\def\norm#1{\left\Vert #1\right\Vert }
\global\long\def\R{\mathbb{R}}
 \global\long\def\Rn{\mathbb{R}^{n}}
\global\long\def\tr{\mathrm{Tr}}
\global\long\def\diag{\mathrm{diag}}
\global\long\def\vol{\mathrm{Vol}}
\global\long\def\E{\mathbb{E}}
\global\long\def\time{\mathcal{T}}
\global\long\def\SO{\text{SO}}
\global\long\def\linspan{\mbox{span}}
\global\long\def\eps{\varepsilon}
\global\long\def\ortho{\perp}
\global\long\def\giventhat{\mid}
\global\long\def\RR{\mathbb{R}}
\global\long\def\card#1{\lvert#1\rvert}
\global\long\def\oproj{\mbox{proj}}
\global\long\def\proj{\pi}
\global\long\def\poly{\mbox{poly}}
\global\long\def\suchthat{\mathrel{:}}
\global\long\def\abs#1{\left|#1\right|}
\global\long\def\ZZ{\mathbb{Z}}
\global\long\def\width{\mbox{width}}
\global\long\def\otilde{\widetilde{O}}
\global\long\def\indicfunc#1{1_{#1}}

%% file: intro.tex
\section{Introduction}

Minimizing a convex function over a convex set is a fundamental problem
with many applications. The problem stands at the forefront of polynomial-time
tractability and its study has lead to the development of numerous
general algorithmic techniques. In recent years, improvements to important
special cases (e.g., maxflow) have been closely related to ideas and
improvements for the general problem \cite{christiano2011electrical,sherman2013nearly,madry2013navigating,kelner2014almost,lee2013new,lee2014path,lee2015efficient,lee2015faster,DBLP:conf/stoc/Sherman17}. 

Here we consider the very general setting where the objective function
and feasible region are both presented only as oracles that can be
queried, specifically an \emph{evaluation oracle} for the function
and a \emph{membership oracle} for the set. We study the problem of
minimizing a convex function over a convex set provided only these
oracles as well as bounds $0<r<R$ and a point $x_{0}\in K$ s.t.
$B(x_{0},r)\subseteq K\subseteq B(x_{0},R)$ where $B(x_{0},r)$ is
the ball of radius $r$ centered at $x_{0}\in\R^{n}$. 

It is well-known that with a stronger \emph{separation} oracle for
the set (and subgradient oracle for the function), this problem can
be solved with $\tilde{O}(n)$ oracle queries using any of \cite{Va96,BV04,lee2015faster}
or with $\tilde{O}(n^{2})$ queries by the classic ellipsoid algorithm
\cite{GLS}. Moreover, it is known that the problem can be solved
with only evaluation and membership oracles through reductions shown
by Grötschel, Lovasz and Schrijver in their classic book \cite{GLS}.
However, the reduction in \cite{GLS} appears to take at least $n^{10}$
calls to the membership oracle. This has been improved using the random
walk method and simulated annealing to $n^{4.5}$ \cite{KV06,LV06}
and  \cite{AbernethyH16}  provides further improvements of up to
a factor of $\sqrt{n}$ for more structured convex sets. 

Our main result in this paper is an algorithm that minimizes a convex
function over a convex set using only $\tilde{O}(n^{2})$ membership
and evaluation queries. Interestingly, we obtain this result by first
showing that we can implement a separation oracle for a convex set
and a subgradient oracle for a function using only $\tilde{O}(n)$
membership queries (Section~\ref{sec:From-Membership-to}) and then
using the known reduction from optimization to separation (Section~\ref{sec:From-Separation-to}).
We state the result informally below. The formal statements, which
allow an approximate membership oracle, are Theorem~\ref{thm:separate_set}
and Theorem~\ref{thm:conv_opt}. 
\begin{thm}
Let $K$ be a convex set specified by a membership oracle, a point
$x_{0}\in\R^{n}$, and numbers $0<r<R$ such that $B(x_{0},r)\subseteq K\subseteq B(x_{0},R)$.
For any convex function $f$ given by an evaluation oracle and any
$\epsilon>0$, there is a randomized algorithm that computes a point
$z\in B(K,\epsilon)$ such that. 
\[
f(z)\le\min_{x\in K}f(x)+\epsilon\left(\max_{x\in K}f(x)-\min_{x\in K}f(x)\right)
\]
with constant probability using $O\left(n^{2}\log^{O(1)}\left(\frac{nR}{\epsilon r}\right)\right)$
calls to the membership oracle and evaluation oracle and $O(n^{3}\log^{O(1)}\left(\frac{nR}{\epsilon r}\right))$
total arithmetic operations. 
\end{thm}
Protasov \cite{Protasov1996} gives an algorithm for approximately
minimizing a convex function defined over an explicit convex body
in $\R^{n}$, using $O(n^{2}\log(n)\log(1/\eps))$ function evaluations,
a logarithmic factor higher. Unfortunately, each iteration of his
algorithm requires computing the convex hull, John ellipsoid and centroid
of a set maintained by the algorithm, thereby making a very large
number of calls to the membership oracle (in \cite{Protasov1996}
the focus is on the number of function calls and it is assumed that
the set is known to the algorithm). We remark that using the main
idea from our algorithm, Protasov's method can be made more efficient,
resulting in oracle complexity that is only a logarithmic factor higher,
although still with a much higher arithmetic complexity than the results
of this paper.

In Section~\ref{sec:Reductions-Between-Oracles} we consider to consequences
of our main result. In \cite{GLS}, the authors describe five basic
problems over convex sets as oracles (OPTimization, SEParation, MEMbership,
VIOLation and VALidity) and give polynomial-time reductions between
them. With our new algorithm, several of these reductions become significantly
more efficient, as summarized in Theorem~\ref{thm:set-reductions}.
In discussing these reductions, it is natural to introduce oracles
for convex functions. The relationships between set oracles and function
oracles are described in Lemma~\ref{lem:from_f_to_Kf} and those
between function oracles in Lemma~\ref{lem:grad_f_star_opt_G}. Figure\ref{fig:rel}
illustrates these relationships and is an updated version of Figure
4.1 from \cite{GLS}. We suspect that the resulting complexities of
reductions are all asymptotically optimal in terms of the dimension,
up to logarithmic factors.\,

\begin{figure}
\begin{minipage}[t]{0.55\columnwidth}%
\begin{tikzpicture}[node distance = 1.8cm]
\tikzstyle{block} = [rectangle, draw, text width=6.5em, text centered, rounded corners, minimum height=4em]
\tikzstyle{line} = [draw, -tonew,arrowhead=0.085cm]
\tikzstyle{dot_line} = [draw, -tonew,dashed,arrowhead=0.085cm]
   \node [block] (opt) {\scriptsize $OPT(K) = \partial {\delta_K}*$};
   \node [block, below of=opt, node distance=1.8cm] (val) {\scriptsize $VAL(K) = {\delta_K}*$};
   \node [block, right of=opt, node distance=3cm] (sepp) {\scriptsize $SEP(K) = \partial{\delta_K}$};
   \node [block, below of=sepp, node distance=1.8cm] (mem) {\scriptsize $MEM(K) = {\delta_K}$};
   \node [block, left of=opt, node distance=3cm] (viol) {\scriptsize $VIOL(K)$};
   \node [below of=viol, node distance=1.5cm] {$\tilde{O}(1) \xrightarrow{\quad\quad}$};
   \node [below of=viol, node distance=2.1cm] {$\tilde{O}(n) \xdashrightarrow{\quad\quad}$}; 
   \path [line] (opt) -- (viol);
   \path [line] (viol) -- (opt);
   \path [line,transform canvas={xshift=-0.25cm}] (sepp) -- (mem);
   \path [line,transform canvas={xshift=-0.25cm}] (opt) -- (val);
   \path [dot_line,transform canvas={xshift=0.25cm}] (mem) -- (sepp);
   \path [dot_line,transform canvas={xshift=0.25cm}] (val) -- (opt);
   \path [dot_line,transform canvas={yshift=-0.25cm}] (opt) -- (sepp);
   \path [dot_line,transform canvas={yshift=0.25cm}] (sepp) -- (opt);  
\end{tikzpicture}%
\end{minipage}\hspace{0.1\columnwidth}%
\begin{minipage}[t]{0.35\columnwidth}%
\begin{tikzpicture}[node distance = 1.8cm]
\tikzstyle{block} = [rectangle, draw, text width=6.5em, text centered, rounded corners, minimum height=4em]
\tikzstyle{line} = [draw, -tonew,arrowhead=0.085cm]
\tikzstyle{dot_line} = [draw, -tonew,dashed,arrowhead=0.085cm]
   \node [block] (opt) {\scriptsize $GRAD(f^*) = \partial f^*$};
   \node [block, below of=opt, node distance=1.8cm] (val) {\scriptsize $EVAL(f^*) = f^*$};
   \node [block, right of=opt, node distance=3cm] (sepp) {\scriptsize $GRAD(f) = \partial{f}$};
   \node [block, below of=sepp, node distance=1.8cm] (mem) {\scriptsize $EVAL(f) = {f}$};
   \path [line,transform canvas={xshift=-0.25cm}] (sepp) -- (mem);
   \path [line,transform canvas={xshift=-0.25cm}] (opt) -- (val);
   \path [dot_line,transform canvas={xshift=0.25cm}] (mem) -- (sepp);
   \path [dot_line,transform canvas={xshift=0.25cm}] (val) -- (opt);
   \path [dot_line,transform canvas={yshift=-0.25cm}] (opt) -- (sepp);
   \path [dot_line,transform canvas={yshift=0.25cm}] (sepp) -- (opt);  
\end{tikzpicture}%
\end{minipage}

\caption{The left diagram illustrates the relationships of the five oracles
defined in \cite{GLS}. The right diagram illustrates the relationships
of oracles for a convex function $f$ and its convex conjugate $f^{*}$.
\label{fig:rel}}
\end{figure}
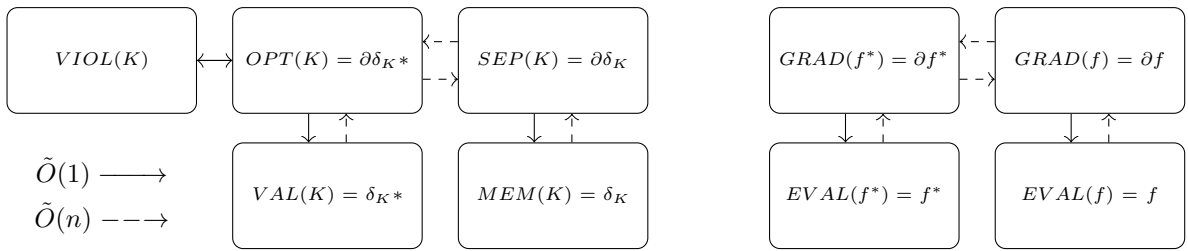

%% file: prelim.tex
\section{Preliminaries}

Here we introduce notation and terminology. Our conventions are chosen
for simplicity and consistency with Grötschel, Lovasz and Schrijver
\cite{GLS}. We use $[n]\defeq\{1,...,n\}$. For a convex function
$f:\R^{n}\rightarrow\R$ and $x\in\R^{n}$ we use $\partial f(x)$
to denote the set of subgradients of $f$ at $x$. For $p>1$,$\delta\geq0$,
and $K\subseteq\R^{n}$ we let 
\[
B_{p}(K,\delta)\defeq\left\{ x\in\R^{n}:\exists y\in K\text{ such that }\norm{x-y}_{p}\leq\delta\right\} 
\]
denote the set of points at distance at most $\delta$ from $K$ in
$\ell_{p}$ norm. For convenience we overload notation and for $x\in\R^{n}$
let $B_{p}(x,\delta)\defeq B_{p}(\{x\},\delta)$ denote the ball of
radius $\delta$ around $x$. We also let 
\[
B_{p}(K,-\delta)\defeq\{x\in\R^{n}:B_{p}(x,\delta)\subseteq K\}
\]
denote the set of points such that the $\delta$ radius balls centered
on them are contained in $K$. In this notation, whenever $p$ is
omitted it is assumed that $p=2$. Furthermore, for any set $K\subseteq\R^{n}$
we let $\indicfunc K$ denote a function from $\R^{n}$ to $\R\cup\{+\infty\}$
such that $\indicfunc K(x)=0$ if $x\in K$ and $\indicfunc K(x)=\infty$
otherwise.

\subsection{Oracles for Convex Sets }

\label{subsec:GLS} Here we provide the five basic oracles for a convex
set, $K\subseteq\R^{n}$, defined by Grötschel, Lovasz and Schrijver
\cite{GLS}. We simplify notation slightly by using the same parameter,
$\delta>0$, to bound both the approximation error and the probability
of failure. 
\begin{defn}[Optimization Oracle (OPT)]
\label{def:OPT} Queried with a unit vector $c\in\Rn$ and a real
number $\delta>0$, with probability $1-\delta$, the oracle either

\begin{itemize}
\item finds a vector $y\in\Rn$ such that $y\in B(K,\delta)$ and $c^{T}x\leq c^{T}y+\delta$
for all $x\in B(K,-\delta)$, or
\item asserts that $B(K,-\delta)$ is empty.
\end{itemize}
We let $\text{OPT}_{\delta}(K)$ be the time complexity of this oracle. 
\end{defn}
\begin{defn}[Violation Oracle (VIOL)]
\label{def:VIOL}Queried with a unit vector $c\in\Rn$, a real number
$\gamma$ and a real number $\delta>0$, with probability $1-\delta$,
the oracle either

\begin{itemize}
\item asserts that $c^{T}x\leq\gamma+\delta$ for all $x\in B(K,-\delta)$,
or
\item finds a vector $y\in B(K,\delta)$ with $c^{T}y\geq\gamma-\delta$.
\end{itemize}
We let $\text{VIOL}_{\delta}(K)$ be the time complexity of this oracle. 
\end{defn}
\begin{defn}[Validity Oracle (VAL)]
\label{def:VAL} Queried with a unit vector $c\in\Rn$, a real number
$\gamma$, and a real number $\delta>0$, with probability $1-\delta$,
the oracle either

\begin{itemize}
\item asserts that $c^{T}x\leq\gamma+\delta$ for all $x\in B(K,-\delta)$,
or
\item asserts that $c^{T}x\geq\gamma-\delta$ for some $x\in B(K,\delta)$.
\end{itemize}
We let $\text{VAL}_{\delta}(K)$ be the time complexity of this oracle.
\end{defn}
\begin{defn}[Separation Oracle (SEP)]
\label{def:SEP} Queried with a vector $y\in\Rn$ and a real number
$\delta>0$, with probability $1-\delta$, the oracle either

\begin{itemize}
\item assert that $y\in B(K,\delta)$, or
\item find a unit vector $c\in\Rn$ such that $c^{T}x\leq c^{T}y+\delta$
for all $x\in B(K,-\delta)$.
\end{itemize}
We let $\text{SEP}_{\delta}(K)$ be the time complexity of this oracle.
\end{defn}
\begin{defn}[Membership Oracle (MEM)]
\label{def:MEM} Queried with a vector $y\in\Rn$ and a real number
$\delta>0$, with probability $1-\delta$, either

\begin{itemize}
\item assert that $y\in B(K,\delta)$, or
\item assert that $y\notin B(K,-\delta)$.
\end{itemize}
We let $\text{MEM}_{\delta}(K)$ be the time complexity of this oracle.
\end{defn}

\subsection{Oracles for Convex Functions}

Let $f$ be a function from $\Rn$ to $\R\cup\{+\infty\}$. Recall
that the dual function $f^{*}$is the convex (Fenchel) conjugate of
$f$, defined as 
\[
f^{*}(y)=\sup_{x\in\R^{n}}\langle y,x\rangle-f(x).
\]
In particular $f^{*}(0)=\inf f$. We will use the following two oracles
for functions.
\begin{defn}[Evaluation Oracle (EVAL)]
\label{def:EVAL} Queried with a vector $y$ with $\norm y_{2}\leq1$
and real number $\delta>0$ the oracle finds an extended real number
$\alpha$ such that
\begin{equation}
\min_{x\in B(y,\delta)}f(x)-\delta\leq\alpha\leq\max_{x\in B(y,\delta)}f(x)+\delta.\label{eq:eval_guarantee}
\end{equation}
We let $\text{EVAL}_{\delta}(f)$ be the time complexity of this oracle.
\end{defn}
\begin{defn}[Subgradient Oracle (GRAD)]
\label{def:GRAD} Queried with a vector $y$ with $\norm y_{2}\leq1$
and real numbers $\delta>0$, the oracle outputs an extended real
number $\alpha$ satisfying (\ref{eq:eval_guarantee}) and a vector
$c\in\Rn$ such that 
\begin{equation}
\alpha+c^{T}(x-y)<\max_{z\in B(x,\delta)}f(z)+\delta\text{ for all }x\in\Rn\label{eq:grad_guarantee}
\end{equation}
We let $\text{GRAD}_{\delta}(f)$ be the time complexity of this oracle.
\end{defn}

%% file: mem_to_sep.tex
\section{From Membership to Separation \label{sec:From-Membership-to}}

In this section, we show that how to implement a separation oracle
for a convex set using only a nearly linear number of queries to a
membership oracle. We divide the construction into two steps. In Section~\ref{sec:alg:separation_convex_Lip},
we show how to compute an approximate subgradient of a Lipshitz convex
function via finite differences. Using this, in Section~\ref{sec:alg:separation_convex_set}
we compute an approximate separating hyperplane for a convex set using
a membership oracle for the set. The algorithms are stated in Algorithm~\ref{alg:sep_set}
and Algorithm~\ref{alg:separateFunc}.

\begin{algorithm2e}[t]\label{alg:sep_set}

\caption{$\mathtt{Separate}_{\varepsilon,\rho}(K,x)$}

\SetAlgoLined

\textbf{Require:} $B_{2}(0,r)\subset K\subset B_{2}(0,R)$.

\uIf{$\text{MEM}_{\varepsilon}(K)$ asserts that $x\in B(K,\epsilon)$}{

\textbf{Output:} $x\in B(K,\varepsilon)$.

}\ElseIf{$x\notin B_{2}(0,R)$}{

\textbf{Output:} the half space $\{y:0\geq\left\langle y-x,x\right\rangle \}$.

}

Let $\kappa=R/r$, $\alpha_{x}(d)=\max_{d+\alpha x\in K}\alpha$ and
$h_{x}(d)=-\alpha_{x}(d)\norm x_{2}$. 

The evaluation oracle of $\alpha_{x}(d)$ can be implemented via binary
search and $\text{MEM}_{\varepsilon}(K)$.

Compute $\tilde{g}=\mathtt{SeparateConvexFunc}(h_{x},0,r_{1},4\varepsilon)$
with $r_{1}=n^{1/6}\varepsilon^{1/3}R^{2/3}\kappa^{-1}$ and the evaluation
oracle of $\alpha_{x}(d)$.

\textbf{Output:} the half space
\[
\left\{ y:\frac{50}{\rho}n^{7/6}R^{2/3}\kappa\varepsilon^{1/3}\geq\left\langle \tilde{g},y-x\right\rangle \right\} 
\]

\end{algorithm2e}

The output of the algorithm for separation is a halfspace that approximately
contains $K$ and the input point $x$ is close to its bounding hyperplane.
It uses a call to a an subgradient function given below. 

\begin{algorithm2e}[t] \label{alg:separateFunc}

\caption{$\mathtt{SeparateConvexFunc}(f,x,r_{1},\varepsilon)$}

\SetAlgoLined

\textbf{Require:} $r_{1}>0$, $\norm{\partial f(z)}_{\infty}\leq L$
for any $z\in B_{\infty}(x,2r_{1})$.

Set $r_{2}=\sqrt{\frac{\varepsilon r_{1}}{\sqrt{n}L}}$.

Sample $y\in B_{\infty}(x,r_{1})$ and $z\in B_{\infty}(y,r_{2})$
independently and uniformly at random.

\For{$i=1,2,\cdots,n$ }{

Let $\alpha_{i}$ and $\beta_{i}$ denote the end points of the interval
$B_{\infty}(y,r_{2})\cap\{z+se_{i}:s\in\R\}$.

Set $\tilde{g}_{i}=\frac{f(\beta_{i})-f(\alpha_{i})}{2r_{2}}$ where
we compute $f$ with $\varepsilon$ additive error.

}

\textbf{Output} $\tilde{g}$ as the approximate subgradient of $f$
at $x$.

\end{algorithm2e}

\subsection{Separation for Lipschitz Convex Function}

\label{sec:alg:separation_convex_Lip}

Here we show how to construct a separation oracle for Lipschitz convex
function given an evaluation oracle. Our construction is motivated
by the following property of convex functions proved by Bubeck and
Eldan \cite[Lem 6]{bubeck2015multi}: for any Lipschitz convex function
$f$, there exists a small ball $B$ such that $f$ restricted on
$B$ is close to a linear function. By a small modification of their
proof, we show this property in fact holds for almost every small
ball (Lemma~\ref{lem:convex_almost_flat}). This can be viewed as
a quantitative version of the Alexandrov theorem for Lipschitz convex
functions.

Leveraging this powerful fact, our algorithm is simple: we compute
a random partial difference in each coordinate to get a subgradient
(Algorithm~\ref{alg:separateFunc}). We prove that as long as the
box we compute over sufficiently small and the additive error in our
evaluation oracle is sufficiently small, this yields an accurate separation
oracle in expectation (Lemma~\ref{lem:separate_conv_func}). We then
obtain high probability bounds using Markov's inequality.

In our analysis we use $\ast$ to denote the convolution operator,
i.e. $(f\ast g)(x)=\int_{\Rn}f(y)g(x-y)dy$.
\begin{lem}
\label{lem:convex_almost_flat} For any $0<r_{2}\leq r_{1}$ and twice
differentiable convex function $f$ defined on $B_{\infty}(x,r_{1}+r_{2})$
with $\norm{\nabla f(z)}_{\infty}\leq L$ for any $z\in B_{\infty}(x,r_{1}+r_{2})$
we have 
\[
\E_{y\in B_{\infty}(x,r_{1})}\E_{z\in B_{\infty}(y,r_{2})}\norm{\nabla f(z)-g(y)}_{1}\leq n^{3/2}\frac{r_{2}}{r_{1}}L
\]
where $g(y)$ is the average of $\nabla f$ over $B_{\infty}(y,r_{2})$.
\end{lem}
\begin{proof}
Let $h=\frac{1}{(2r_{2})^{n}}f\ast1_{B_{\infty}(0,r_{2})}$. Integrating
by parts, we have that
\[
\int_{B_{\infty}(x,r_{1})}\Delta h(y)dy=\int_{\partial B_{\infty}(x,r_{1})}\left\langle \nabla h(y),n(y)\right\rangle dy
\]
where $\Delta h(y)=\sum_{i}\frac{d^{2}h}{dx_{i}^{2}}(y)$ and $n(y)$
is the normal vector on $\partial B_{\infty}(x,r_{1})$ the boundary
of the box $B_{\infty}(x,r_{1})$, i.e. standard basis vectors. Since
$f$ is $L$-Lipschitz with respect to $\norm{\cdot}_{\infty}$ so
is $h$, i.e. $\norm{\nabla h(z)}_{\infty}\leq L$. Hence, we have
that
\[
\E_{y\in B_{\infty}(x,r_{1})}\Delta h(y)\leq\frac{1}{(2r_{1})^{n}}\int_{\partial B_{\infty}(x,r_{1})}\norm{\nabla h(y)}_{\infty}\norm{n(y)}_{1}dy\leq\frac{1}{(2r_{1})^{n}}\cdot2n(2r_{1})^{n-1}\cdot L=\frac{nL}{r_{1}}.
\]
By the definition of $h$, we have that
\begin{equation}
\E_{y\in B_{\infty}(x,r_{1})}\E_{z\in B_{\infty}(y,r_{2})}\Delta f(z)=\E_{y\in B_{\infty}(x,r_{1})}\Delta h(y)\leq\frac{nL}{r_{1}}.\label{eq:Laplacian_bound}
\end{equation}
Let $\omega_{i}(z)=\left\langle \nabla f(z)-g(y),e_{i}\right\rangle $
for all $i\in[n]$. Since $\int_{B_{\infty}(y,r_{2})}\omega_{i}(z)dz=0$,
the Poincare inequality for a box (see e.g. \cite{steinerberger2015sharp})
shows that
\[
\int_{B_{\infty}(y,r_{2})}\left|\omega_{i}(z)\right|dz\leq r_{2}\int_{B_{\infty}(y,r_{2})}\norm{\nabla\omega_{i}(z)}_{2}dz.
\]
Since $f$ is convex, we have that $\norm{\nabla^{2}f(z)}_{F}\leq\tr\nabla^{2}f(z)=\Delta f(z)$
and hence 
\[
\sum_{i\in[n]}\norm{\nabla\omega_{i}(z)}_{2}=\sum_{i\in[n]}\norm{\nabla^{2}f(z)e_{i}}_{2}\leq\sqrt{n}\norm{\nabla^{2}f(z)}_{F}\leq\sqrt{n}\Delta f(z).
\]
Using this with $\norm{\nabla f(z)-g(y)}_{1}=\sum_{i}|\omega_{i}(z)|$,
we have that
\[
\int_{B_{\infty}(y,r_{2})}\norm{\nabla f(z)-g(y)}_{1}dz\leq\sqrt{n}r_{2}\int_{B_{\infty}(y,r_{2})}\Delta f(z)dz.
\]
Combining with the inequality (\ref{eq:Laplacian_bound}) yields the
result.
\end{proof}
\begin{lem}
\label{lem:separate_conv_func} Given $r_{1}>0$. Let $f$ be a convex
function on $B_{\infty}(x,2r_{1})$. Suppose that $\norm{\partial f(z)}_{\infty}\leq L$
for any $z\in B_{\infty}(x,2r_{1})$. Also, assume that we can compute
function f with $\varepsilon$ additive error with $\varepsilon\leq r_{1}\sqrt{n}L$.Let
$\tilde{g}=\mathtt{SeparateConvexFunc}(f,x,r_{1},\varepsilon)$. Then,
there is random variable $\zeta\geq0$ with $\E\zeta\leq3\sqrt{\frac{L\varepsilon}{r_{1}}}n^{5/4}$
such that 
\[
f(q)\geq f(x)+\left\langle \tilde{g},q-x\right\rangle -\zeta\norm{q-x}_{\infty}-4nr_{1}L\text{ for all }q\in\Omega.
\]
\end{lem}
\begin{proof}
By limiting argument, we assume that $f$ is twice differentiable.

First, we assume that we can compute $f$ exactly, namely $\varepsilon=0$.
Fix $i\in[n]$. Let $g(y)$ is the average of $\nabla f$ over $B_{\infty}(y,r_{2})$.
Then, we have that
\begin{align*}
\E_{z}\left|\tilde{g}_{i}-g(y)_{i}\right| & =\E_{z}\left|\frac{f(\beta_{i})-f(\alpha_{i})}{2r_{2}}-g(y)_{i}\right|\\
 & \leq\E_{z}\frac{1}{2r_{2}}\int\left|\frac{df}{dx_{i}}(z+se_{i})-g(y)_{i}\right|ds\\
 & =\E_{z}\left|\frac{df}{dx_{i}}(z)-g(y)_{i}\right|
\end{align*}
where we used that both $z+se_{i}$ and $z$ are uniform distribution
on $B_{\infty}(y,r_{2})$ in the last line. Hence, we have
\[
\E_{z}\norm{\tilde{g}-\nabla f(z)}_{1}\leq\E_{z}\norm{\nabla f(z)-g(y)}_{1}+\E_{z}\norm{\tilde{g}-g(y)}_{1}\leq2\E_{z}\norm{\nabla f(z)-g(y)}_{1}.
\]
Now, applying the convexity of $f$ yields that
\begin{align*}
f(q) & \geq f(z)+\left\langle \nabla f(z),q-z\right\rangle \\
 & =f(z)+\left\langle \tilde{g},q-x\right\rangle +\left\langle \nabla f(z)-\tilde{g},q-x\right\rangle +\left\langle \nabla f(z),x-z\right\rangle \\
 & \geq f(x)+\left\langle \tilde{g},q-x\right\rangle -\norm{\nabla f(z)-\tilde{g}}_{1}\norm{q-x}_{\infty}-\norm{\nabla f(z)}_{\infty}\norm{x-z}_{1}.
\end{align*}
Now, $\norm{\nabla f(z)}_{\infty}\leq L$ and $\norm{x-z}_{1}\leq n\cdot\norm{x-z}_{\infty}\leq2n(r_{1}+r_{2})$
by assumption. Furthermore, we can apply Lemma~\ref{lem:convex_almost_flat}
to bound $\norm{\nabla f(z)-\tilde{g}}_{1}$ and use that $r_{2}=\sqrt{\frac{\varepsilon r_{1}}{\sqrt{n}L}}\leq r_{1}$
to get
\[
f(q)\geq f(x)+\left\langle \tilde{g},q-x\right\rangle -\zeta\norm{q-x}_{\infty}-4nr_{1}L
\]
with $\E\zeta\leq2n^{3/2}\frac{r_{2}}{r_{1}}L$.

Since we only compute $f$ up to $\varepsilon$ additive error, it
introduces $\frac{\varepsilon}{r_{2}}$ additive error into $\tilde{g}_{i}$.
Hence, we instead have that
\[
\E\zeta\leq2n^{3/2}\frac{r_{2}}{r_{1}}L+\frac{\varepsilon n}{r_{2}}.
\]
Putting $r_{2}=\sqrt{\frac{\varepsilon r_{1}}{\sqrt{n}L}}$, we get
the bound.
\end{proof}

\subsection{Separation for Convex Set\label{sec:alg:separation_convex_set}}

Throughout this subsection, let $K\subseteq\R^{n}$ be a convex set
that contains $B_{2}(0,r)$ and is contained in $B_{2}(0,R)$. Given
some point $x\notin K$, we wish to separate $x$ from $K$ using
a membership oracle. To do this, we reduce this problem to computing
an approximate subgradient of a Lipschitz convex function, called
$h_{x}(d)$, the ``height'' of a point $d$ in the direction of
$x$. We let $\alpha_{x}(d)=\max_{d+\alpha x\in K}\alpha$ and define
$h_{x}(d)=-\alpha_{x}(d)\norm x_{2}$. Note that $d+\alpha_{x}(d)x$
is the last point on the line passing through $d$ and $d+x$ that
is in $K$ and that $-h_{x}(d)$ is the $\ell_{2}$ distance from
this point to $d$.
\begin{lem}
\label{lem:h_convex} $h_{x}(d)$ is convex on $K$.
\end{lem}
\begin{proof}
Let $d_{1},d_{2}\in K$ and $\lambda\in[0,1]$ be arbitrary. Now $d_{1}+\alpha_{x}(d_{1})x\in K$
and $d_{2}+\alpha_{x}(d_{2})x\in K$ and consequently,
\[
\left[\lambda d_{1}+(1-\lambda)d_{2}\right]+\left[\lambda\cdot\alpha_{x}(d_{1})+(1-\lambda)\cdot\alpha_{x}(d_{2})\right]x\in K\,.
\]
Therefore, if we let $d\defeq\lambda d_{1}+(1-\lambda)d_{2}$ we see
that $\alpha_{x}(d)\geq\lambda\cdot\alpha_{x}(d_{1})+(1-\lambda)\cdot\alpha_{x}(d_{2})$
and $h_{x}(\lambda d_{1}+(1-\lambda d_{2})\leq\lambda h_{x}(d_{1})+\lambda h_{x}(d_{2}).$
\end{proof}
\begin{lem}
\label{lem:h_Lip} $h_{x}$ is $\frac{R+\delta}{r-\delta}$ Lipschitz
over points in $B_{2}(0,\delta)$ for $\delta<r$. 
\end{lem}
\begin{proof}
Let $d_{1},d_{2}$ be arbitrary points in $B(0,\delta)$. We wish
to upper bound $\left|h_{x}(d_{1})-h_{x}(d_{2})\right|$ in terms
of $\norm{d_{1}-d_{2}}_{2}$. We assume without loss of generality
that $\alpha_{x}(d_{1})\geq\alpha_{x}(d_{2})$ and therefore
\[
\left|h_{x}(d_{1})-h_{x}(d_{2})\right|=\left|\alpha_{x}(d_{1})\norm x_{2}-\alpha_{x}(d_{2})\norm x_{2}\right|=\left(\alpha_{x}(d_{1})-\alpha_{x}(d_{2})\right)\norm x_{2}\,.
\]
Consequently, it suffices to lower bound $\alpha_{x}(d_{2})$. We
split the analysis into two cases.

Case 1: $\norm{d_{2}-d_{1}}_{2}\leq r-\delta$. We consider the point
$d_{3}=d_{1}+\frac{d_{2}-d_{1}}{\lambda}$ with $\lambda=\norm{d_{2}-d_{1}}_{2}/(r-\delta)$.
Note that 
\[
\norm{d_{3}}_{2}\leq\norm{d_{1}}_{2}+\frac{1}{\lambda}\norm{d_{2}-d_{1}}_{2}\leq\delta+\frac{1}{\lambda}\norm{d_{2}-d_{1}}_{2}\leq r.
\]
Hence, $d_{3}\in K$. Since $\lambda\in[0,1]$ and $K$ is convex,
we have that $\lambda\cdot d_{3}+(1-\lambda)\cdot\left[d_{1}+\alpha_{x}(d_{1})x\right]\in K$.
Now, we note that
\[
\lambda\cdot d_{3}+(1-\lambda)\cdot\left[d_{1}+\alpha_{x}(d_{1})x\right]=d_{2}+(1-\lambda)\cdot\alpha_{x}(d_{1})x
\]
and this shows that 
\[
\alpha_{x}(d_{2})\geq(1-\lambda)\cdot\alpha_{x}(d_{1})=\left(1-\frac{\norm{d_{2}-d_{1}}_{2}}{r-\delta}\right)\cdot\alpha_{x}(d_{1}).
\]
Since $d_{1}+\alpha_{x}(d_{1})x\in K\subset B_{2}(0,R)$, we have
that $\alpha_{x}(d_{1})\cdot\norm x_{2}\leq R+\delta$ and hence
\[
\left|h_{x}(d_{1})-h_{x}(d_{2})\right|=(\alpha_{x}(d_{1})-\alpha_{x}(d_{2}))\cdot\norm x_{2}\leq\alpha_{x}(d_{1})\cdot\norm x_{2}\frac{\norm{d_{2}-d_{1}}_{2}}{r-\delta}\leq\frac{R+\delta}{r-\delta}\norm{d_{2}-d_{1}}_{2}.
\]

Case 2: $\norm{d_{2}-d_{1}}_{2}\geq r-\delta$. Since $0\geq h_{x}(d_{1}),h_{x}(d_{2})\geq-R-\delta$,
we have that 
\[
\left|h_{x}(d_{1})-h_{x}(d_{2})\right|\leq R+\delta\leq\frac{R+\delta}{r-\delta}\norm{d_{2}-d_{1}}_{2}\,.
\]
In either case we have that 
\[
\left|h_{x}(d_{1})-h_{x}(d_{2})\right|\leq\frac{R+\delta}{r-\delta}\norm{d_{2}-d_{1}}_{2}
\]
yielding the desired result.
\end{proof}
\begin{lem}
\label{lem:separate_set} Let $K$ be a convex set satisfying $B_{2}(0,r)\subset K\subset B_{2}(0,R)$.
Given any $0<\rho<1$ and $0\leq\varepsilon\leq r$. With probability
$1-\rho$, $\mathtt{Separate}_{\varepsilon,\rho}(K,x)$ outputs a
half space that contains $K$. 
\end{lem}
\begin{proof}
When $x\notin B_{2}(0,R)$, the algorithm outputs a valid separation
for $B_{2}(0,R)$. For the rest of the proof, we assume $x\notin B(K,-\varepsilon)$
(due to the membership oracle) and $x\in B_{2}(0,R)$.

By Lemma \ref{lem:h_convex} and Lemma \ref{lem:h_Lip}, $h_{x}$
is convex with Lipschitz constant $3\kappa$ on $B_{2}(0,\frac{r}{2})$.
By our assumption on $\varepsilon$ and our choice of $r_{1}$, we
have that $B_{\infty}(0,2r_{1})\subset B_{2}(0,\frac{r}{2})$. Hence,
we can apply Lemma \ref{lem:separate_conv_func} to get that
\begin{align}
h_{x}(y) & \geq h_{x}(0)+\left\langle \tilde{g},y\right\rangle -\zeta\norm y_{\infty}-12nr_{1}\kappa\label{eq:hxy-1}
\end{align}
for any $y\in K$. Note that $-\frac{x}{\kappa}\in K$ and $h_{x}(-\frac{x}{\kappa})=h_{x}(0)-\frac{1}{\kappa}\norm x_{2}$.
Hence, we have
\[
h_{x}(0)-\frac{1}{\kappa}\norm x_{2}=h_{x}(-\frac{1}{\kappa}x)\geq h_{x}(0)+\left\langle \tilde{g},-\frac{1}{\kappa}x\right\rangle -\frac{1}{\kappa}\zeta\norm x_{\infty}-12nr_{1}\kappa.
\]
Therefore, we have
\begin{equation}
\left\langle \tilde{g},x\right\rangle \geq\norm x_{2}-\zeta\norm x_{\infty}-12nr_{1}\kappa^{2}.\label{eq:lowerbound_g}
\end{equation}
Now, we note that $x\notin B(K,-\varepsilon)$. Using that $B(0,r)\subset K$,
we have $(1-\frac{\varepsilon}{r})K\subset B(K,-\varepsilon)$. Hence,
\[
h_{x}(0)\geq-\left(1-\frac{\varepsilon}{r}\right)\norm x_{2}\geq-\norm x_{2}+\varepsilon\kappa.
\]
Therefore, we have
\[
h_{x}(0)+\left\langle \tilde{g},x\right\rangle \geq-\zeta\norm x_{\infty}-12nr_{1}\kappa^{2}-\varepsilon\kappa
\]
Combining this with (\ref{eq:hxy-1}), we have that 
\begin{align*}
h_{x}(y) & \geq\left\langle \tilde{g},y-x\right\rangle -\zeta\norm y_{\infty}-\zeta\norm x_{\infty}-12nr_{1}\kappa-12nr_{1}\kappa^{2}-\varepsilon\kappa\\
 & \geq\left\langle \tilde{g},y-x\right\rangle -2\zeta R-24nr_{1}\kappa^{2}-\varepsilon\kappa
\end{align*}
for any $y\in K$. Recall from Lemma \ref{lem:separate_conv_func}
that $\zeta$ is a positive random scalar independent of $y$ satisfying
$\E\zeta\leq3\sqrt{\frac{12\kappa\varepsilon}{r_{1}}}n^{5/4}.$ For
any $y\in K$, we have that $h_{x}(y)\leq0$ and hence $\tilde{\zeta}\geq\left\langle \tilde{g},y-x\right\rangle $
where $\tilde{\zeta}$ is a random scalar independent of $y$ satisfying
\begin{align*}
\E\tilde{\zeta} & \le6\sqrt{\frac{12\kappa\varepsilon}{r_{1}}}n^{5/4}R+24nr_{1}\kappa^{2}+\varepsilon\kappa\\
 & \leq45n^{7/6}R^{2/3}\varepsilon^{1/3}\kappa+\varepsilon\kappa\\
 & \leq50n^{7/6}R^{2/3}\varepsilon^{1/3}\kappa
\end{align*}
where we used $0\leq\varepsilon\leq r$ at the end. The result then
follows from this and Markov inequality. 
\end{proof}
\begin{thm}
\label{thm:separate_set} Let $K$ be a convex set satisfying $B_{2}(0,1/\kappa)\subset K\subset B_{2}(0,1)$.
For any $0\leq\eta<\frac{1}{2}$, we have that
\[
\text{SEP}_{\eta}(K)\leq O\left(n\log\left(\frac{n\kappa}{\eta}\right)\right)\text{MEM}_{(\eta/n\kappa)^{O(1)}}(K).
\]
\end{thm}
\begin{proof}
First, we bound the running time. Note that the bottleneck is to compute
$h_{x}$ with $\varepsilon$ additive error. Since $-O(1)\leq h_{x}(y)\leq0$
for all $y\in B_{2}(0,O(1))$, one can compute $h_{x}(y)$ by binary
search with $O(\log(1/\delta))$ calls to the membership oracle.

Next, we check that $\mathtt{Separate}_{\delta,\rho}(K,x)$ is indeed
a separation oracle. Note that $\tilde{g}$ may not be an unit vector
and we need to re-normalize the $\tilde{g}$ by $1/\norm{\tilde{g}}_{2}$.
So, we need to a lower bound $\norm{\tilde{g}}_{2}$. 

From (\ref{eq:lowerbound_g}) and our choice of $r_{1}$, if $\delta\leq\frac{\rho^{3}}{10^{6}n^{6}\kappa^{6}}$,
then we have that
\begin{align*}
\left\langle \tilde{g},x\right\rangle  & \geq\norm x_{2}-\zeta\norm x_{\infty}-12nr_{1}\kappa^{2}\geq\frac{r}{4}.
\end{align*}
Hence, we have that $\norm{\tilde{g}}_{2}\geq\frac{1}{4\kappa}$.
Therefore, this algorithm is a separation oracle with error $\frac{200}{\rho}n^{7/6}\kappa^{2}\delta^{1/3}$
and failure probability $O(\rho+\log(1/\delta)\delta)$.

\[
\text{SEP}_{\Omega(\max(n^{7/6}\kappa^{2}\delta^{1/3}/\rho+\rho+\log(1/\delta)\delta)}(K)\leq O(\log(1/\delta))\text{MEM}_{\delta}(K).
\]
Setting $\rho=\sqrt{n^{7/6}\kappa^{2}\delta^{1/3}}$ and $\delta=\Theta\left(\frac{\eta^{6}}{n^{7/2}\kappa^{6}}\right)$,
we have that 
\[
\text{SEP}_{\eta}(K)\leq O(\log(\frac{n\kappa}{\eta}))\text{MEM}_{\eta^{6}/(n^{7/2}\kappa^{6})}(K).
\]
\end{proof}

%% file: sep_to_opt.tex
\section{From Separation to Optimization \label{sec:From-Separation-to}}

Once we have a separation oracle, our running times follow by applying
a recent convex optimization algorithm by Lee, Sidford and Wong \cite{lee2015faster}.
Previous algorithms also achieved $\tilde{O}(n)$ oracle complexity,
but needed a higher polynomial number of arithmetic operations. We
remark that the theorem stated in \cite{lee2015faster} is slightly
more general then the one we give below, but since we only need to
minimize linear functions over convex sets, we state a simplified
version here.
\begin{thm}[Theorem 42 of \cite{lee2015faster} Rephrased]
\label{thm:conv_opt} Let $K$ be a convex set satisfying $B_{2}(0,r)\subset K\subset B_{2}(0,1)$
and let $\kappa=1/r$. For any $0<\varepsilon<1$, with probability
$1-\varepsilon$, we can compute $x\in B(K,\varepsilon)$ such that
\[
c^{T}x\leq\min_{x\in K}c^{T}x+\varepsilon\norm c_{2}
\]
with an expected running time of 
\[
O\left(n\text{SEP}_{\delta}(K)\log\left(\frac{n\kappa}{\varepsilon}\right)+n^{3}\log^{O(1)}\left(\frac{n\kappa}{\varepsilon}\right)\right),
\]
where $\delta=(\frac{\varepsilon}{n\kappa})^{\Theta(1)}$. In other
words, we have that
\[
\text{OPT}_{\varepsilon}(K)=O\left(n\text{SEP}_{(\frac{\varepsilon}{n\kappa})^{\Theta(1)}}(K)\log\left(\frac{n\kappa}{\varepsilon}\right)+n^{3}\log^{O(1)}\left(\frac{n\kappa}{\varepsilon}\right)\right).
\]
\end{thm}

%% file: equivalency.tex
\section{Reductions Between Oracles \label{sec:Reductions-Between-Oracles}}

In this section, we provide all other reductions among oracles defined
in Section~\ref{subsec:GLS}. To simplify notation we assume the
convex set is contained in the unit ball and convex function is defined
on the unit ball. This can be done without loss of generality by scaling
and shifting. 

We remark that it is known that OPT and VIOL are equivalent up to
the cost of a binary search. 
\begin{lem}[Equivalence between OPT and VIOL]
\label{def:OPT_VIOL}Given a convex set $K$ contained in the unit
ball, we have that $\text{VIOL}_{\delta}(K)\leq\text{OPT}_{\delta}(K)$
and $\text{OPT}_{\delta}(K)\leq O(\log(1+\frac{1}{\delta}))\cdot\text{VIOL}_{\delta}(K)$
for any $\delta>0$.
\end{lem}
Hence, we ignore $\text{VIOL}$ for the remainder of this section. 

\subsection{Relationships between Set oracles and Function Oracles}

Next, to handle all these relationships efficiently, we find it convenient
to instead look at oracles on convex functions and connect them to
set oracles. For this purpose we note the following simple relationship
between $\text{MEM}(K)$ and $\text{EVAL}(1_{K})$ and between $\text{SEP}(K)$
and $\text{GRAD}(1_{K})$.
\begin{lem}[$\text{MEM}(K)$ and $\text{SEP}(K)$ are membership and subgradient
oracle of $1_{K}$]
\label{lem:SEP_K_GRAD_1K} For any convex set $K\subseteq\R^{n}$,
we have that $\text{MEM}_{\delta}(K)=\text{EVAL}_{\delta}(1_{K})$
and $\text{SEP}_{\delta}(K)=\text{GRAD}_{\delta}(1_{K})$ for any
$\delta>0$.
\end{lem}
Next, we note that the relationship between $\text{VAL}(K)$ and $\text{EVAL}(1_{K}^{*})$
and between $\text{OPT}(K)$ and $\text{GRAD}(1_{K}^{*})$.
\begin{lem}[$\text{VAL}(K)$ and $\text{OPT}(K)$ are membership and subgradient
oracle of $1_{K}^{*}$]
\label{lem:OPT_K_GRAD_1K} Given a convex set $K$. Suppose that
$B(\vec{0},r)\subset K\subset B(\vec{0},1)$ and let $\kappa=1/r$.
For any $\delta>0$, we have that 
\begin{itemize}
\item $\text{VAL}_{\delta}(K)\leq\text{EVAL}_{\delta}(1_{K}^{*})$ and $\text{EVAL}_{\delta}(1_{K}^{*})\leq O(\log(\kappa/\delta))\cdot\text{VAL}_{\Omega(\delta/(\kappa\log(1/\delta))}(K)$.
\item $\text{OPT}_{\delta}(K)\leq\text{GRAD}_{\delta/4}(1_{K}^{*})$ and
$\text{GRAD}_{\delta}(1_{K}^{*})\leq\text{OPT}_{\delta/(3+\kappa)}(K)$.
\end{itemize}
where the oracle for $1_{K}^{*}$ is only defined on the unit ball.
\end{lem}
\begin{proof}
For the first inequality, to implement the validity oracle, we need
to compute $\beta$ such that
\begin{equation}
\max_{x\in B(K,-\delta)}c^{T}x-\delta\leq\beta\leq\min_{x\in B(K,\delta)}c^{T}x+\delta\label{eq:val_req}
\end{equation}
for any unit vector $c$ and $\delta>0$. We note that
\[
\min_{x\in B(c,\frac{\delta}{R})}1_{K}^{*}(x)\geq1_{K}^{*}(c)-\delta=\max_{x\in K}c^{T}x-\delta\geq\max_{x\in B(K,-\delta)}c^{T}x.
\]
Therefore, (\ref{eq:eval_guarantee}) shows that the output $\alpha$
by $\text{EVAL}_{\delta}(1_{K}^{*})$ with input $-c$ satisfies $\max_{x\in B(K,-\delta)}c^{T}x\leq\alpha+\delta$.
Similarly, we have that $\min_{x\in B(K,\delta)}c^{T}x\geq\alpha-\delta$.
Thus, the output of $\text{EVAL}_{\delta}(1_{K}^{*})$ satisfies the
condition (\ref{eq:val_req}). Hence, we have that $\text{VAL}_{\delta}(K)\leq\text{EVAL}_{\delta}(1_{K}^{*})$.

For the second inequality, to implement the evaluation oracle of $1_{K}^{*}$,
we need to compute $1_{K}^{*}(c)=\max_{x\in K}c^{T}x$ for any vector
$c$ with $\norm c_{2}\leq1$. Using that $B(0,r)\subset K$, we have
$(1-\frac{\delta}{r})K\subset B(K,-\delta)$. Hence, we have that
\[
\max_{x\in B(K,-\delta)}c^{T}x\geq(1-\frac{\delta}{r})\max_{x\in K}c^{T}x\geq\max_{x\in K}c^{T}x-\kappa\delta.
\]
On the other hand, we have that
\[
\max_{x\in B(K,\delta)}c^{T}x\leq\max_{x\in K}c^{T}x+\delta.
\]
Hence, by binary search on $\gamma$, $\text{VAL}_{\delta}(K)$ allows
us to estimate $\max_{x\in K}c^{T}x$ up to $2(2+\kappa)\delta$ additive
error.

For the third inequality, to implement the optimization oracle, we
let $c$ be the vector we want to optimize. Let $x$ be the output
of $\text{GRAD}_{\eta}(1_{K}^{*})$ on input $c$. Using (\ref{eq:grad_guarantee})
and (\ref{eq:eval_guarantee}), we have that
\[
\min_{z\in B(c,\eta)}1_{K}^{*}(z)+x^{T}(d-c)<\max_{z\in B(d,\eta)}1_{K}^{*}(z)+2\eta
\]
for any vector $d$. Since $1_{K}^{*}$ is $R$-Lipschitz, we have
that 
\[
1_{K}^{*}(c)+x^{T}(d-c)<1_{K}^{*}(d)+4\eta.
\]
Putting $d=0$, we have
\[
\max_{x\in K}c^{T}x=1_{K}^{*}(c)\leq c^{T}x+4\eta.
\]
Setting $\eta=\delta/4$, we see that $x$ is a maximizer of $\max_{x\in K}c^{T}x$
up to $\delta$ additive error.

For the fourth inequality, to implement the subgradient oracle, we
let $c$ be the point we want to compute the subgradient such that
$\norm c_{2}\leq1$. Let $y$ be the output of $\text{OPT}_{\delta}(K)$
with input $c$. Since $(1-\frac{\delta}{r})K\subset B(K,-\delta)$,
we have that
\[
\max_{x\in K}c^{T}x\leq c^{T}y+\delta+\kappa\delta.
\]
Therefore, 
\[
c^{T}y+(d-c)^{T}y\leq\max_{x\in K}c^{T}x+\delta+(d-c)^{T}y\leq d^{T}y+(2+\kappa)\delta\leq\max_{x\in K}d^{T}x+(3+\kappa)\delta.
\]
Let $\alpha=c^{T}y$. Since $y\in B(K,\delta)$ and satisfies the
guarantee of optimization oracle, $\alpha$ satisfies (\ref{eq:eval_guarantee})
with additive error $\delta$. Furthermore, we note that 
\[
\alpha+y^{T}(d-c)\leq1_{K}^{*}(d)+(3+\kappa)\delta.
\]
Hence, it satisfies (\ref{eq:grad_guarantee}) with additive error
$(3+\kappa)\delta$.
\end{proof}
\begin{lem}
\label{lem:from_f_to_Kf} Given a convex function $f:B_{n}\rightarrow[0,1]$,
let $K_{f}=\{(\frac{x}{2},\frac{t}{4})\text{ such that }x\in B(0,1)\text{ and }f(x)\leq t\leq2\}$.
Then, 
\begin{itemize}
\item $\text{MEM}_{\delta}(K_{f})\leq\text{EVAL}_{\delta/10}(f)$ and $\text{EVAL}_{\delta}(f)\leq O(\log(1/\delta))\text{MEM}_{\Omega(\delta/\log(1/\delta))}(K_{f})$.
\item $\text{SEP}_{\delta}(K_{f})\leq\text{GRAD}_{\delta/10}(f)$ and $\text{GRAD}_{\delta}(f)\leq O(\log(1/\delta))\text{SEP}_{\Omega(\delta/\log(1/\delta))}(K_{f})$.
\item $\text{GRAD}_{\delta}(f^{*})\leq\text{OPT}_{\delta/6}(K_{f})$.
\end{itemize}
\end{lem}
\begin{proof}
The first two sets of reductions are clear.

For the last one, to implement the subgradient oracle, we let $c$
be the point we want to compute the subgradient such that $\norm c_{2}\leq1$.
Let $(y,t')$ be the output of $\text{OPT}_{\delta}(K_{f})$ with
input $(c,-1)$. Since $(1-4\delta)K_{f}\subset B(K_{f},-\delta)$,
we have that
\[
\max_{(x,t)\in K_{f}}(c^{T}x-t)\leq c^{T}y-t'+5\delta.
\]
Since $(y,t')\in B(K_{f},\delta)$, for any vector $d$, we have that
\[
(c^{T}y-t')+(d-c)^{T}y\leq\max_{(x,t)\in K_{f}}(c^{T}x-t)+(d-c)^{T}y\leq d^{T}y-t'+5\delta\leq\max_{(x,t)\in K_{f}}(d^{T}x-t)+6\delta.
\]
Let $\alpha=c^{T}y-t'$. Since $y\in B(K,\delta)$ and satisfies the
guarantee of optimization oracle, $\alpha$ is a good enough approximation
of $f^{*}(c)$. Furthermore, we note that 
\[
\alpha+y^{T}(d-c)\leq\max_{(x,t)\in K_{f}}d^{T}x+5\delta=f^{*}(d)+6\delta
\]
Hence, it satisfies (\ref{eq:grad_guarantee}) with additive error
$6\delta$.
\end{proof}

\subsection{Relationships Between Convex Function Oracles}

Due to the equivalences above, we can focus on the more general problem:
the relationships between
\begin{itemize}
\item $\text{EVAL}_{\delta}(f)$, $\text{GRAD}_{\delta}(f)$, $\text{EVAL}_{\delta}(f^{*})$,
$\text{GRAD}_{\delta}(f^{*})$.
\end{itemize}
\begin{lem}
\label{lem:grad_f_star_opt_G}Given a convex function $f$ defined
on unit ball with value between $0$ and $1$. For any $0\leq\delta\leq\frac{1}{2}$,
we have that
\begin{itemize}
\item $\text{EVAL}_{\delta}(f)\leq\text{GRAD}_{\delta}(f)\leq O(n\log^{2}(\frac{n}{\delta}))\text{MEM}_{(\delta/n)^{O(1)}}(K_{f})\le O(n\log^{2}(\frac{n}{\delta}))\text{EVAL}_{(\delta/n)^{O(1)}}(f)$
\item $\text{GRAD}_{\delta}(f^{*})\leq\text{OPT}_{\delta/6}(K_{f})$ and
\begin{align*}
\text{OPT}_{\delta/6}(K_{f}) & \le O\left(n\text{SEP}_{(\delta/n)^{O(1)}}(K_{f})\log\left(\frac{n}{\delta}\right)+n^{3}\log^{O(1)}\left(\frac{n}{\delta}\right)\right)\\
 & \leq O\left(n\log\left(\frac{n}{\delta}\right)\cdot\text{GRAD}_{(\delta/n)^{O(1)}}(f)+n^{3}\log^{O(1)}\left(\frac{n}{\delta}\right)\right).
\end{align*}
\end{itemize}
\end{lem}
\begin{proof}
The bound $\text{EVAL}_{\delta,\eta}(f)\leq\text{GRAD}_{\delta,\eta}(f)$
is immediate from definition. 

To bound $\text{GRAD}(f)$ by $\text{EVAL}(f)$, we use Lemma \ref{lem:from_f_to_Kf}
and get that
\[
\text{GRAD}_{\delta}(f)\leq O(\log(\delta^{-1}))\text{SEP}_{\Omega(\delta/\log(\delta^{-1}))}(K_{f}).
\]
Next, we note that $B(0,0.1)\subset K_{f}\subset B(0,1)$. Hence,
Theorem \ref{thm:separate_set} shows that 
\[
\text{SEP}_{\delta}(K_{f})\leq O(n\log(\frac{n}{\delta}))\text{MEM}_{(\delta/n)^{O(1)}}(K_{f}).
\]
Hence, we have that 
\[
\text{GRAD}_{\delta}(f)\leq O(n\log^{2}(\frac{n}{\delta}))\text{MEM}_{(\delta/n)^{O(1)}}(K_{f}).
\]
Applying Lemma \ref{lem:from_f_to_Kf} again, we have the result.

To bound $\text{GRAD}(f^{*})$ by $\text{GRAD}(f)$, we again use
Lemma \ref{lem:from_f_to_Kf} and Theorem~\ref{thm:conv_opt} to
get
\begin{align*}
\text{GRAD}_{\delta}(f^{*}) & \leq\text{OPT}_{\delta/6}(K_{f})\leq O\left(n\text{SEP}_{(\delta/n)^{O(1)}}(K_{f})\log\left(\frac{n}{\delta}\right)+n^{3}\log^{O(1)}\left(\frac{n}{\delta}\right)\right)\\
 & \leq O\left(n\text{GRAD}_{\delta}(f)\log\left(\frac{n}{\delta}\right)+n^{3}\log^{O(1)}\left(\frac{n}{\delta}\right)\right).
\end{align*}
\end{proof}

\subsection{Relationships Between Convex Set Oracles}
\begin{thm}
\label{thm:set-reductions}For any convex set $K$ such that $B(0,1/\kappa)\subset K\subset B(0,1)$,
for any $0<\delta<\frac{1}{2}$, we have that
\begin{enumerate}
\item $\text{VIOL}_{\delta}(K)\leq\text{OPT}_{\delta}(K)$ and $\text{OPT}_{\delta}(K)\leq O(\log(1+\frac{1}{\delta}))\cdot\text{VIOL}_{\Theta(\delta/\log(1/\delta))}(K)$.
\item $\text{MEM}_{\delta}(K)\leq\text{SEP}_{\delta}(K)$ and $\text{SEP}_{\delta}(K)\leq O(n\log(\frac{n\kappa}{\delta}))\cdot\text{MEM}_{(\delta/n\kappa)^{O(1)}}(K).$
\item $\text{VAL}_{\delta}(K)\leq\text{OPT}_{\delta}(K)$ and $\text{OPT}_{\delta}(K)\leq O(n\log^{3}(\frac{n\kappa}{\delta}))\cdot\text{VAL}_{(\delta/n\kappa)^{O(1)}}(K).$
\item $\text{OPT}_{\delta}(K)=O\left(n\log\left(\frac{n\kappa}{\delta}\right)\cdot\text{SEP}_{(\delta/(n\kappa))^{O(1)}}(K)+n^{3}\log^{O(1)}\left(\frac{n\kappa}{\delta}\right)\right)$.
\item $\text{SEP}_{\delta}(K)=O\left(n\log\left(\frac{n}{\delta}\right)\cdot\text{OPT}_{(\delta/(n\kappa))^{O(1)}}(K)+n^{3}\log^{O(1)}\left(\frac{n}{\delta}\right)\right)$.
\end{enumerate}
\end{thm}
\begin{proof}
(1) follows from Lemma \ref{def:OPT_VIOL}. (2) follows from Theorem
\ref{thm:separate_set}. (4) follows from Theorem \ref{thm:conv_opt}.

For (3), we use Lemma \ref{lem:OPT_K_GRAD_1K}, \ref{lem:grad_f_star_opt_G}
and \ref{lem:OPT_K_GRAD_1K} to get
\begin{align*}
\text{OPT}_{\delta}(K) & \leq\text{GRAD}_{\delta/4}(1_{K}^{*})\leq O(n\log(\frac{n}{\delta}))\text{EVAL}_{(\delta/n)^{O(1)}}(1_{K}^{*})\\
 & \leq O(n\log^{2}(\frac{n\kappa}{\delta}))\text{VAL}_{(\delta/n\kappa)^{O(1)}}(K)
\end{align*}
where we used that $1_{K}^{*}$ is a function between $0$ and $1$.

For (5), we use Lemma \ref{lem:SEP_K_GRAD_1K}, \ref{lem:grad_f_star_opt_G}
and \ref{lem:OPT_K_GRAD_1K}
\begin{align*}
\text{SEP}_{\delta}(K) & =\text{GRAD}_{\delta}(1_{K})\leq O\left(n\log\left(\frac{n}{\delta}\right)\cdot\text{GRAD}_{(\delta/n)^{O(1)}}(1_{K}^{*})+n^{3}\log^{O(1)}\left(\frac{n}{\delta}\right)\right)\\
 & \leq O\left(n\log\left(\frac{n}{\delta}\right)\cdot\text{OPT}_{(\delta/(n\kappa))^{O(1)}}(K)+n^{3}\log^{O(1)}\left(\frac{n}{\delta}\right)\right)
\end{align*}
where we used that $1_{K}^{*}$ is a function between $0$ and $1$.
\end{proof}